\documentclass[11pt]{article}

\usepackage{fullpage}
\usepackage{amsmath,amsthm,amssymb}
\usepackage{algorithm}
\usepackage{algpseudocode}
\usepackage{tabularx}

\newtheorem{theorem}{Theorem}[section]
\newtheorem{lemma}[theorem]{Lemma}
\newtheorem{corollary}[theorem]{Corollary}

\newtheorem{proposition}[theorem]{Proposition}

\makeatletter
\newcommand{\multiline}[1]{%
	\begin{tabularx}{\dimexpr\linewidth-\ALG@thistlm}[t]{@{}X@{}}
		#1
	\end{tabularx}
}
\makeatother

\begin{document}

\title{Fast Byzantine Gathering with Visibility in Graphs}

\author{Avery Miller and
Ullash Saha\\University of Manitoba, Winnipeg, MB, Canada}

\date{}
\maketitle              % typeset the header of the contribution
\begin{abstract}
We consider the gathering task by a team of $m$ synchronous mobile robots in a graph of $n$ nodes. Each robot has an identifier (ID) and runs its own deterministic algorithm, i.e., there is no centralized coordinator. We consider a particularly challenging scenario: there are $f$ Byzantine robots in the team that can behave arbitrarily, and even have the ability to change their IDs to any value at any time. There is no way to distinguish these robots from non-faulty robots, other than perhaps observing strange or unexpected behaviour. The goal of the gathering task is to eventually have all non-faulty robots located at the same node in the same round. It is known that no algorithm can solve this task unless there at least $f+1$ non-faulty robots in the team. In this paper, we design an algorithm that runs in polynomial time with respect to $n$ and $m$ that matches this bound, i.e., it works in a team that has exactly $f+1$ non-faulty robots. In our model, we have equipped the robots with sensors that enable each robot to see the subgraph (including robots) within some distance $H$ of its current node. We prove that the gathering task is solvable if this visibility range $H$ is at least the radius of the graph, and not solvable if $H$ is any fixed constant.
\end{abstract}

\section{Introduction}
Mobile robots play a vital role in real-life applications such as military surveillance, search-and-rescue, environmental monitoring, transportation, mining, infrastructure protection, and autonomous vehicles. In networks, the robots/agents move from one location to another to collectively complete a task, and might all need to meet at one location in order to share information or start their next task. Therefore, gathering becomes a fundamental problem for mobile robots in networks.

Gathering is hard to accomplish even in a fault-free system, as the robots may not have any planned location where to meet, nor any initial information about the topology of the network. Moreover, in a distributed system, each robot runs its own deterministic algorithm to make decisions, i.e., there is no centralized coordinator. We want a deterministic algorithm that can be run by each robot, and eventually, they will gather at a single node which is not fixed in advance. Additionally, we consider a particularly challenging scenario in which some of the robots are Byzantine: such robots do not follow our installed algorithm and can behave arbitrarily. We can think of these robots as malicious robots in our system, i.e., they have been compromised by outsiders/hackers, and, knowing the algorithm we intend to run, they can behave in ways that attempt to mislead the non-faulty robots into making incorrect decisions. Moreover, non-faulty robots do not know which of the robots (or even how many of the robots) are Byzantine, because all robots look identical. We might face this type of scenario in real-world applications when attackers try to disrupt the normal behavior of systems, so algorithms that are resilient to such attacks are very useful. 

The relative number of non-faulty robots versus Byzantine robots is an essential factor in solving this problem. If there are many Byzantine robots compared to the number of non-faulty robots, then the behaviour of the Byzantine robots can be very influential. As shown in previous work \cite{dieudonne2014gathering}, a team that contains $f$ Byzantine robots cannot solve gathering if the number of non-faulty robots is less than $f+1$. The challenge, and the goal of our work, is to provide an efficient gathering algorithm that works when this bound is met, i.e., when the number of non-faulty robots is exactly $f+1$. We provide such an algorithm in a model in which each robot is endowed with sensors that allow them to see all nodes and robots within a fixed distance $H$ of its current location, where $H$ is at least the radius of the network. We also prove an impossibility result which shows that no algorithm can solve gathering in this model if $H$ is any fixed constant (i.e., independent of any graph parameter). It's important to note that this impossibility result does not contradict previous results \cite{bouchard2016byzantine,bouchard2018byzantineICALP,dieudonne2014gathering,Tsuch} that provide gathering algorithms with no visibility, as those algorithms make assumptions about additional information known to the robots (such as bounds on the network size, or on the number of Byzantine robots) or make assumptions about additional features such as authenticated whiteboards at the nodes.

\subsection{Model and Definitions}\label{Model}
We consider a team of $m$ robots that are initially placed at arbitrary nodes of an undirected connected graph $G = (V, E)$. We denote by $n$ the number of nodes in the graph, i.e., $n = |V|$. The nodes have no labels. At each node $v$, the incident edges are labeled with port numbers $0,\ldots,deg(v)-1$ in an arbitrary way, where $deg(v)$ represents the degree of node $v$. The two endpoints of an edge need not be labeled with the same port number. 

For any two nodes $v,w$, the \emph{distance between $v$ and $w$}, denoted by $d(v,w)$, is defined as the length of a shortest path between $v$ and $w$. The \emph{eccentricity of a node $v$}, denoted by $ecc(v)$, is the maximum distance from $v$ to any other node, i.e., $ecc(v) = \max_{w \in V}\{d(v,w)\}$. The \emph{radius} of a graph, denoted by $R$, is defined as the minimum eccentricity taken over all nodes, i.e., $R = \min_{v \in V}\{ecc(v)\}$.

The team of $m$ robots contains $f$ Byzantine robots and $m-f$ non-faulty robots. Each robot $\alpha$ has a distinct identifier (ID) $l_\alpha$, and it knows its own ID. The Byzantine and non-faulty robots look identical, i.e., there is no way to distinguish them other than perhaps noticing strange or unexpected behaviour. All robots have unbounded memory, i.e., they can remember all information that they have previously gained during their algorithm's execution. We describe the differences between the two types of robots below.

\subsubsection{Properties of non-faulty robots.} The non-faulty robots have no initial information about the size or topology of the graph, and they have no information about the number of Byzantine robots. A non-negative integer parameter $H$ defines the \emph{visibility range} of each robot, which we describe in Partial Snapshot below. Each non-faulty robot executes a synchronous deterministic algorithm: in each round, each robot performs one Look-Compute-Move sequence, i.e., it performs the following three operations in the presented order.
	
	\begin{enumerate}
		\item {\textbf{The Look operation:}} A non-faulty robot $\alpha$ located at a node $v$ at the start of round $t$ gains information from two types of view.  
		\begin{itemize}
			\item \textit{Local View:} Robot $\alpha$ can see the degree of node $v$ and the port numbers of its incident edges. It can also see any other robots located at $v$ at the start of round $t$, along with their ID numbers. %the same node can see each other. However, the robots traversing the same edge do not notice each other.
			
			\item \textit{Partial Snapshot View:} Robot $\alpha$ sees the subgraph consisting of all nodes, edges, and port numbers that belong to paths of length at most $H$ that have $v$ as one endpoint. Also, for each node $w$ in this subgraph, robot $\alpha$ sees the list of all IDs of the robots occupying $w$ at the start of round $t$. %More precisely, the robot gets a snapshot $S$ that returns the current configuration of the subgraph $S(G)$.
		\end{itemize}
		
		\item {\textbf{The Compute operation:}} Using the information gained during all previous Look operations, a robot $\alpha$ located at a node $v$ deterministically chooses a value from the set $\{null,0,\ldots,deg(v)-1\}$. In particular, it chooses $null$ if it decides that it will stay at its current node $v$, and it chooses a value $p \in \{0,\ldots,deg(v)-1\}$ if it decides to move to the neighbour of node $v$ that is the other endpoint of the incident edge labeled with port number $p$. 
		%decides whether it will stay at its current node $v$ or if it will move to a neighbour of $v$ in this round, If it decides that it will move the port-number/edge, through which it will move to one of the adjacent nodes. The robot may choose a null port-number, meaning it will stay at its current node.
		
		\item {\textbf{The Move operation:}} A robot $\alpha$ located at a node $v$ performs the action that it chose during the Compute operation. In particular, it does nothing if it chose value $null$, and otherwise, it moves towards a neighbour $w$ of $v$ along the incident edge labeled with the chosen port number $p$, and it arrives at $w$ at the start of the next round. It sees the port number that it uses to enter node $w$. There is no restriction of how robots move along an edge, i.e., multiple robots may traverse an edge simultaneously, in either direction.
		%Robots can move along edges in both directions, and multiple robots can walk along the same edge in opposite directions, i.e., “pass by each other”. However, the robot does not move when the target node is the current node.
	\end{enumerate}
	All non-faulty robots wake up at the same time and perform their Look-Compute-Move sequences synchronously in every round.
	%cycles in fully-synchronous rounds. In the fully-synchronous model, robots are activated at the same time and perform the operations instantaneously. More precisely, in our model, time proceeds in synchronized rounds where each round contains a Look-Compute-Move cycle. All robots participate in every round. 
	%\smallskip{}
	
\subsubsection{Properties of the Byzantine robots.} 
	%Moreover, they cannot distinguish the Byzantine robots because all robots look the same.
	%We assume that there are $f$ strongly Byzantine robots in the graph $G$. 
	We assume that a centralized adversary controls all of the Byzantine robots. This adversary has complete knowledge of the algorithm being executed by the non-faulty robots, and can see the entire network and the positions of all robots at all times. In each round, the adversary can make each Byzantine robot move to an arbitrary neighbouring node. Further, we assume that the faulty robots are \emph{strongly Byzantine}, which means that the adversary can change the ID of any Byzantine robot at any time (in contrast, a \emph{weakly Byzantine} robot would have a fixed ID during the entire execution). 
	%Strongly Byzantine implies that the robots can change their IDs at any round of the execution. Moreover, they look the same as the good robots. In any round, a Byzantine robot can move to any arbitrary adjacent node without following the applied algorithm.     

\subsubsection{Problem Statement.}
Assume that $m$ robots are initially placed at nodes of a network, where $f$ of the robots are strongly Byzantine. The robots synchronously execute a deterministic distributed algorithm. Eventually, all non-faulty robots must terminate their algorithm in the same round, and at termination, all non-faulty robots must be located at the same node.

%\section{Literature Review and Our contribution} \label{Literature}
\subsection{Related Work}
The study of algorithms for mobile robots is extensive, as evidenced by a recent survey \cite{MobileRobotsSurveyBook}. The Gathering problem has been investigated thoroughly under a wide variety of model assumptions, as summarized in \cite{BhagatSurvey,DefagoSurvey,FloccSurvey} for continuous models and in \cite{CicerSurvey,PelcSurvey} for discrete models. Of particular interest to our current work are discrete models where the robots are located in a network, have some amount of visibility beyond its own position \cite{Barra,Barri,Chal,Fisch,Hsiang}, and where faults may occur \cite{ChalFaults,Oosh,PelcCrash}. 
%
%With regards to visibility, the authors of \cite{Chal} showed that equipping mobile robots with 1-hop visibility strictly increases their computational power as they can solve new instances of Graph Exploration. In \cite{Barri}, the authors show that Uniform Scattering is possible in grid graphs by a team of anonymous robots that move asynchronously, have limited visibility range, have constant-sized memory, and are equipped with a compass. The authors of \cite{Barra,Hsiang} also study the Uniform Scattering problem with limited visibility and limited memory, but in connected subsets of grid graphs and when there is a restriction on where the agents can enter the environment. In \cite{Fisch}, the authors consider the task of having all robots converge to a 2x2 grid within any grid graph, where the robots are anonymous, synchronous, have limited visibility range, and have no memory.
%
%The Gathering problem has also been studied in networks under various fault scenarios. Previous work has focused on the case where robots have no visibility, i.e., they cannot see beyond their own current location. In \cite{ChalFaults}, the authors study Gathering when faults can temporarily prevent a robot from moving in some rounds. The authors of \cite{Oosh} provide self-stabilizing gathering algorithms for two synchronous mobile robots, which means that synchronous two-robot gathering can be solved under any type of transient fault. In \cite{PelcCrash}, Gathering was solved when robots moved asynchronously and were subject to crash faults.

Most relevant to our current work are the results about Gathering in networks when some of the robots can be Byzantine \cite{bouchard2016byzantine,bouchard2018byzantineICALP,dieudonne2014gathering,Tsuch}. In \cite{Tsuch}, the authors consider weakly Byzantine agents and add authenticated whiteboards to the model. %Each node has a whiteboard where robots can leave messages: each robot has a dedicated space on each whiteboard that it can write to, and no other robots (even Byzantine ones) can write or erase a space that is not dedicated to them. All spaces on a whiteboard can be read by all robots. 
Additionally, each robot has the ability to write ``signed" messages that authenticate the ID of the writer and whether the message was originally written at the current node. The authors provide an algorithm such that all correct robots gather at a single node in $O(f\cdot |E|)$ rounds, where $f$ is an upper bound on the number of Byzantine robots and $|E|$ is the number of edges in the network.

For the model we consider in our work (but with visibility range 0), the Gathering problem was first considered in \cite{dieudonne2014gathering}. The authors explored the gathering problem under four variants of the model: (i) known size of the graph, weakly Byzantine robots, (ii) known size of the graph, strongly Byzantine robots, (iii) unknown size of the graph, weakly Byzantine robots, and (iv) unknown size of the graph, strongly Byzantine robots. In all cases, the authors assume that the upper bound $f$ on the number of Byzantine robots is known to all non-faulty robots. The authors provided a deterministic polynomial-time algorithms for the two models with weakly Byzantine robots. In the model when the size of the graph is known, their algorithm works for any number of non-faulty robots in the team. Recently, the authors of \cite{hir2020gather} provided a significantly faster algorithm under the assumption that the number of non-faulty robots in the team is at least $4f^2+8f+4$. In \cite{dieudonne2014gathering}, assuming that the size of the graph is unknown and $f$ robots are weakly Byzantine, the authors provide an algorithm that works when the number of non-faulty robots in the team is $f+2$. They prove a matching lower bound in this scenario: no algorithm can solve Gathering if the number of non-faulty robots in the team is less than $f+2$. For the model with strongly Byzantine robots and known graph size, the authors provided a randomized algorithm that guarantees that the agents gather in a finite number of rounds, and with high probability terminates in $n^{cf}$ rounds for some constant $c > 0$. They also provided a deterministic algorithm whose running time is exponential in $n$ and the largest ID belonging to a non-faulty agent. In both cases, the number of non-faulty robots in the team is assumed to be at least $2f+1$. The authors also proved a lower bound for this model: no algorithm can solve Gathering if the number of non-faulty robots in the team is less than $f+1$. Finally, for the model with strongly Byzantine robots and unknown graph size, they provided a deterministic algorithm that works when the number of non-faulty robots in the team is at least $4f+2$. The running time is exponential in $n$ and the largest ID belonging to a non-faulty agent. They also proved a lower bound in this model: no algorithm can solve Gathering if the number of non-faulty robots in the team is less than $f+2$. Subsequent work focused on the case of strongly Byzantine robots and attempted to close the gaps between the known upper and lower bounds on the number of non-faulty robots in the team. This was achieved in \cite{bouchard2016byzantine}, as the authors provided algorithms that work when the number of non-faulty robots in the team are $f+1$ and $f+2$ for the cases of known and unknown graph size, respectively. However, the running times of these algorithms were also exponential in $n$ and the largest ID belonging to a non-faulty agent. 

More recently, the authors of \cite{bouchard2018byzantineICALP} considered a version of the above model that does not assume knowledge of the graph size nor the upper bound $f$ on the number of strongly Byzantine agents. Instead, they considered the amount of initial knowledge as a resource to be quantitatively measured as part of an algorithm's analysis. In this model, they designed an algorithm whose running time is polynomial in $n$ and the number of bits in the smallest ID belonging to a non-faulty agent, where $O(\log\log\log n)$ bits of initial information is provided to all robots. The initial information they provide is the value of $\log\log{n}$, which the algorithm uses as a rough estimate of the graph size. Their algorithm works as long as the number of non-faulty robots in the team is at least $5f^2+6f+2$. They also proved a lower bound on the amount of initial knowledge: for any deterministic polynomial Gathering algorithm that works when the number of non-faulty robots in the team is at least $5f^2+6f+2$ and whose running time is polynomial in $n$ and the number of bits in the smallest ID, the amount of initial information provided to all robots must be at least $\Omega(\log \log \log n)$ bits.      

\subsection{Our Results}
We consider a graph-based model in which each robot has no initial information other than its own ID and has some visibility range $H$. We prove that no algorithm can solve Gathering in the presence of Byzantine robots if $H$ is any fixed constant. We also design an algorithm that solves Gathering in any graph with $n$ nodes containing $m$ robots, $f$ of which are strongly Byzantine, and where each non-faulty robot has visibility range $H$ equal to the radius of the graph (or larger). Our algorithm has the following desirable properties: (1) the number of rounds is polynomial with respect to $n$ and $m$, in contrast to several previous algorithms whose running times are exponential in $n$ and the largest robot ID; (2) it works when the number of non-faulty robots in the team is $f+1$ (or larger), which is optimal due to an impossibility result from \cite{dieudonne2014gathering} that also holds in our model, and significantly improves on the best previous polynomial-time algorithm, which requires at least $5f^2+6f+2$ non-faulty robots; (3) it does not assume any initial global knowledge, in contrast to previous algorithms that assume a known bound on the graph size or on the number of Byzantine robots. Such assumptions might be unrealistic in many applications.

\section{The Algorithm}
First, we define some notation that will be used in the algorithm's description and analysis.
%The following graph-theoretic definitions will be used throughout the description and analysis of our algorithm. Recall that the \emph{eccentricity of a node $v$}, denoted by $ecc(v)$, is the maximum distance from $v$ to any other node, i.e., $ecc(v) = \max_{w \in V}\{d(v,w)\}$. The \emph{radius} of a graph, denoted by $R$, is defined as the minimum eccentricity taken over all nodes, i.e., $R = \min_{v \in V}\{ecc(v)\}$. 
For any graph $G$, the \emph{center of graph $G$} is the set of all nodes that have minimum eccentricity, i.e., all nodes $v \in V(G)$ such that $ecc(v) = R$, and the \emph{center graph of a graph $G$}, denoted by $C(G)$, is defined as the subgraph induced by the center nodes. The following terminology will be used to refer to what a robot $\alpha$ can observe in the Look operation of any round $t$ during the execution of an algorithm. The \emph{local view at a node $v$ for round $t$} is denoted by $Lview(v,t)$, and refers to all of the following information: the degree of $v$, the port numbers of its incident edges, and a list of the IDs of all other robots located at node $v$ at the start of round $t$. The \emph{snapshot view at a node $v$ for round $t$} is denoted by $Sview(v,t)$, and refers to all of the following information: the subgraph consisting of all nodes, edges, and port numbers that belong to paths of length at most $H$ that have $v$ as one endpoint, and, for each node $w$ in this subgraph, the list of IDs of all robots occupying $w$ at the start of round $t$. For any graph $G$, an ID $l$ is called a \emph{singleton ID} if the total number of times that $l$ appears as a robot ID at the nodes of $G$ is exactly 1.

\subsection{Algorithm Description}\label{Algorithm}

In what follows, we assume that the visibility range of a non-faulty robot is at least equal to the radius of the graph, i.e., $H \geq R$.
We also assume that the number of non-faulty robots is at least $f+1$.     

The algorithm's progress can be divided into three parts. The first part makes each non-faulty robot move to a node $v_{max}$ such that the robot's snapshot view from $v_{max}$ contains all the nodes of the network $G$. This is the purpose of our {\bf Find-Lookout} subroutine, which we now describe. Each robot $\alpha$ produces a list of potential nodes in its initial snapshot view where it thinks it might be located, and it does this by comparing its local view with the degree and robot list of each node in its initial snapshot. It cannot be sure of its initial position within its snapshot view since Byzantine robots can forge $\alpha$'s ID and position themselves at other nodes that have the same degree as $\alpha$'s current node. From each guessed initial position, $\alpha$ computes a port sequence of a depth-first traversal of its snapshot view and tries following it in the real network. Since one of the guessed initial positions must be correct, at least one of the depth-first traversals will successfully visit all nodes contained in $\alpha$'s initial snapshot view. Since the visibility range is at least the radius of the network, the robot's initial snapshot view must contain a node in the center of the network $G$, so at least one step of at least one of the traversals will visit a node in the center of $G$. When located at such a node, the robot will see all nodes in the network. So, by counting how many nodes it sees at every traversal step, and keeping track of where it saw the maximum, it can correctly remember and eventually go back to a node $v_{max}$ from which it saw all nodes in the network. See Algorithm \ref{alg:FindLookout} for the pseudocode of Find-Lookout. After returning to $v_{max}$ at the end of Find-Lookout, each robot $\alpha$ constructs a set $P_{\alpha}$ consisting of nodes in its snapshot view that match its local view. These can be thought of as `candidate' locations where $\alpha$ thinks it might actually be located within its snapshot view.
\begin{algorithm}
	\small
	\caption{Find-Lookout, executed by $\alpha$ starting at a node $v$ in round 0}
	\label{alg:FindLookout}
	\begin{algorithmic}[1]
	\State Store the initial snapshot $Sview(v,0)$ in its memory as $S_0$
	\State Determine which nodes in $S_0$ might be its starting location, i.e., compute a set $X$ of nodes $w \in S_0$ where the degree of $w$ and the list of robot ID's at $w$ is the same as $v$'s local view in round 0.
	\For{each $w \in X$}
		\State \multiline{%
			Compute a port sequence $\tau$ corresponding to a depth-first traversal of $S_0$ starting at $w$, and attempt to follow this port sequence in the actual network}
		\State \multiline{%
			In every round of the attempted traversal, take note of the number of nodes seen in the snapshot view, and remember the maximum such number $n_{max}$, a node $v_{max}$ where this maximum was witnessed, the number $m_{max}$ of robots seen in the snapshot when located at $v_{max}$, and the sequence of ports $\tau_{max}$ used to reach $v_{max}$ from the starting location}
		\State Return to the starting node by reversing the steps taken during the attempt
	\EndFor
	\State For the largest $n_{max}$ seen in any of the traversal attempts, go to the corresponding node $v_{max}$ using the sequence $\tau_{max}$
	\end{algorithmic}
\end{algorithm}

The second part of the algorithm ensures that, eventually, there is a robot with a singleton ID that is located in the center of the network $G$. This is the purpose of our {\bf March-to-Center} subroutine, which depends highly on the fact that each robot starts this part of the algorithm at a node $v_{max}$ from which it can see every node in the network. If a robot starts March-to-Center knowing where in its snapshot view it is located (i.e., $|P_\alpha|=1$), then the robot moves directly to the center of $G$: it computes the center of its snapshot, and moves to one of the nodes in the center of this snapshot, which is also the center of the entire network $G$. If all robots do this, then the center of the network will contain a singleton ID, since there are more non-faulty robots than Byzantine robots, and all non-faulty robots have distinct ID's. The difficult case is when a robot $\alpha$ is not sure where in its snapshot it is located at the start of March-to-Center (i.e., $|P_\alpha|>1$). This is because the Byzantine robots can forge $\alpha$'s  ID and position themselves at other nodes with the same degree as $\alpha$'s current node. In this case, $\alpha$ will not move during March-to-Center, and simply watch to see if it can spot any inconsistencies between its local view and its possible starting locations in its snapshot. The key observation, which we will prove, is that at least one of the following must happen in each execution of March-to-Center: there is a robot with a singleton ID located at a node in the center of the network, or, at least one robot sees an inconsistency and narrows down its list of possible starting locations. So, after enough repetitions of March-to-Center, we can guarantee that there will be a robot with a singleton ID that is located in the center of the network. The location of the robot with the smallest such singleton ID is chosen as $v_{target}$ by all non-faulty robots, and this is the place where the robots will eventually gather. See Algorithm \ref{alg:MarchToCenter} for the pseudocode of March-to-Center.

\begin{algorithm}
	\small
	\caption{March-to-Center($P_\alpha$), run by $\alpha$ starting at a node $v$ in round $t$}
	\label{alg:MarchToCenter}
	\begin{algorithmic}[1]
		\If{$|P_\alpha|=1$}
			\State \multiline{%
			Use current snapshot $Sview(v,t)$ to compute a shortest path $\pi$ starting at the node $v_0 \in P_\alpha$ and ending at a node $v_{closest}$ in the center graph $C(Sview(v,t))$ that minimizes the distance $d(v_0,v_{closest})$}
			\State \multiline{%
			Move along the port sequence in $\pi$ and then wait $H-|\pi|$ rounds at $v_{closest}$}
		\Else
			\State \multiline{%
			Wait at current node $v$ for $H$ rounds, and observe every node $v_j \in P_\alpha$ in every snapshot view during the waiting period}
			\State \multiline{%
			If, in any round of the waiting period, there is some $v_j$ that does not have any robot with ID $l_\alpha$, then remove $v_j$ from $P_\alpha$ (as we're not currently located at $v_j$)}
		\EndIf
		\State In both cases, at the end of the waiting period, check if there is a singleton ID in the center graph $C(Sview(v,t+H))$. If there is such a singleton ID, set $v_{target}$ as the node that contains a robot with the smallest singleton ID in $C(Sview(v,t+H))$. Otherwise, $v_{target}$ is set to $null$.
	\end{algorithmic}
\end{algorithm}

The third part of the algorithm gets each robot to successfully move to the target node $v_{target}$, which completes the gathering process. This is the purpose of our Merge subroutine. As above, if a robot starts Merge knowing where in its snapshot view it is located (i.e., $|P_\alpha|=1$), then it can simply compute a sequence of port numbers that leads to $v_{target}$ and follow it. The difficult case is when a robot $\alpha$ is not sure where in its snapshot it is located at the start of Merge (i.e., $|P_\alpha|>1$). In this case, $\alpha$ just tries one node from its list of possibilities, computes a sequence of port numbers that leads to $v_{target}$, and tries to follow it. If it notices any inconsistencies along the way or after it arrives, it deletes the guessed starting node from its list $P_\alpha$. After each Merge, each robot reverses the steps it took during the Merge in order to go back to where it started so that it can run Merge again. Each execution of Merge finishes in one of two ways: all robots have gathered, or, at least one robot has eliminated one incorrect guess about its starting position. So, after a carefully chosen number of repetitions, we can guarantee that the last performed Merge gathers all robots at the same node. See Algorithm \ref{alg:Merge} for the pseudocode of Merge.

\begin{algorithm}
	\small
	\caption{Merge($P_\alpha,v_{target},n$), executed by robot $\alpha$}
	\label{alg:Merge}
	\begin{algorithmic}[1]
		\State Using the snapshot view, determine a shortest path $\pi$ starting at the first node $v_0 \in P_\alpha$ and ending at $v_{target}$.
		\State Attempt to move along the port sequence in $\pi$ to reach $v_{target}$.
		\If{there is a round in which the next port to take along path $\pi$ does not exist in the local view, or, the port used to arrive at the current node is different than the port specified in path $\pi$}
			\State \multiline{%
				Delete $v_0$ from $P_\alpha$, then wait $H-t_{\alpha,Move}$ rounds at the current node, where $t_{\alpha,Move}$ is the number of rounds taken to reach the current node}
		\Else\Comment{the port sequence in $\pi$ was followed with no inconsistency}
			\State \multiline{%
				Wait $H-|\pi|$ rounds at the current node $v$. In each of these rounds $t'$, consider the current snapshot $Sview(v,t')$:}
			\State \indent If the number of nodes in this view is less than $n$, then remove $v_0$ from $P_\alpha$
			\State \indent If $l_\alpha$ is not at $v_{target}$ in $Sview(v,t')$, then remove $v_0$ from $P_\alpha$
			\State \indent \multiline{%
			If the current local view does not match the local view of $v_{target}$ in\\ $Sview(v,t')$ (i.e., a different degree, or a different list of robots),\\
			then remove $v_0$ from $P_\alpha$}
		\EndIf
	\end{algorithmic}
\end{algorithm}

The pseudocode for the complete algorithm, called the H-View-Algorithm, is provided as Algorithm \ref{alg:HView}.

\begin{algorithm}
	\small
	\caption{H-View-Algorithm, run by $\alpha$ starting at node $v$ in round 0}
	\label{alg:HView}
	\begin{algorithmic}[1]
		\State Execute {\bf Find-Lookout()}
		\State Wait at $v_{max}$ until round $x = (m_{max}+2)\cdot n^2_{max}$
		\State In round $x$, create a set $P_\alpha$ consisting of the nodes $w \in Sview(v_{max},x)$ where the degree of $w$ and the list of robot ID's at $w$ are the same as $v_{max}$'s local view in round $x$
		\State Initialize $v_{target}$ to $null$, initialize $phase$ to 1
		\Repeat
			\State Execute {\bf March-to-Center}$(P_\alpha)$
			\State $phase \leftarrow phase + 1$
		\Until{$v_{target} \neq null$}
		\Repeat
			\State Execute {\bf Merge}$(P_\alpha, v_{target}, n_{max})$
			\State Perform the traversals of the previous {\bf Merge} in reverse (returning to $v_{max}$)
			\State $phase \leftarrow phase + 1$
		\Until{$phase > \left\lceil\frac{m_{max}}{2}\right\rceil$}
		\State Execute {\bf Merge}$(P_\alpha,v_{target},n_{max})$
		\State terminate()
	\end{algorithmic}
\end{algorithm}

\subsection{Analysis} \label{Analysis}
We consider three main parts of the algorithm. Our first goal is to show that, immediately after robot $\alpha$ executes {\bf Find-Lookout}, it has moved to a node $v_{max}$ such that the snapshot view from $v_{max}$ contains $n_{max} = n$ nodes and $m_{max} = m$ robots.

\begin{lemma}\label{lemma-1}
	By round $(m+2)\cdot n^2$, each non-faulty robot $\alpha$ is located at a node $v_{max}$ such that the snapshot view at $v_{max}$ contains $n$ nodes and $m$ robots. \end{lemma}

\begin{proof}
	Consider an arbitrary robot $\alpha$'s execution of the H-View-Algorithm starting at a node $v$. First, $\alpha$ computes the set of nodes $w \in Sview(v,0)$ where the degree of $w$ and the list of robot ID's at $w$ is the same as $v$'s local view in round 0. In particular, this means that each such node $w$ contains $\alpha$'s ID $l_\alpha$ in its list of robots. Since at most $f+1$ robots can have ID $l_\alpha$ in round 0 (i.e., $\alpha$ itself and at most $f$ Byzantine robots), we get that the number of nodes $w$ in $Sview(v,0)$ that look the same as $Lview(v,0)$ is at most $f+1$. Consequently, this means that the number of different depth-first traversals attempted by $\alpha$ is at most $f+1$. Each depth-first traversal takes at most $2|E|$ rounds, which is less than $n^2$. Together with the reversal to return back to its starting node, we get that each attempt takes at most $2n^2$ rounds, so all traversals are complete by round $2(f+1)\cdot n^2$. Since one of the computed traversal sequences starts at $\alpha$'s real initial location, it follows that at least one of the traversal attempts visits all nodes in $Sview(v,0)$. By the definition of the network's center and the fact that $H \geq R$, it follows that $Sview(v,0)$ must contain a node that is in the network's center, and we just showed that $\alpha$ necessarily visited all nodes in $Sview(v,0)$. Since the snapshot view at any node in the center of the network contains all of the network's nodes (since $H \geq R$), it follows that $\alpha$ visits at least one node at which the snapshot view contains all $n$ nodes (and contains all $m$ robots). Robot $\alpha$ will save such a node as $v_{max}$, it will set $n_{max} = n$ and $m_{max} = m$, and it will set $\tau_{max}$ to be a port sequence from $v$ to $v_{max}$. The final traversal of the path $\tau_{max}$ to get from $v$ to $v_{max}$ takes at most another $n^2$ rounds, so, in total, $\alpha$ arrives at $v_{max}$ by round $(2f+3)\cdot n^2$. Since the number of non-faulty robots is at least $f+1$, we get that $m \geq 2f+1$, so $f \leq \frac{m-1}{2}$. Thus, $(2f+3)\cdot n^2 \leq (m+2)\cdot n^2$. \end{proof}

The second part of the algorithm consists of the executions of March-to-Center. Our main goal is to prove that, after at most $f+1$ executions of March-to-Center, every robot sets its $v_{target}$ variable to the same non-null value. To this end, we first prove that each execution of March-to-Center by the non-faulty robots is started at the same time, and, at the end of each execution, every robot is located at a node such that its snapshot contains all of the network's nodes. This allows us to conclude that any particular feature seen by one robot can be seen by all other robots at the same time.

\begin{lemma} \label{lemma-2}
	At the end of each execution of March-to-Center by any non-faulty robot $\alpha$, the robot resides at some node $v$ such that its snapshot view contains all the nodes of $G$. \end{lemma}

\begin{proof}
	We consider the two cases in the description of March-to-Center. We note that, at the end of each execution of March-to-Center by a non-faulty robot $\alpha$, either $\alpha$ is at the node $v_{max}$ where it started the execution, or, it is at a node $v_{closest}$ which is defined to be in $C(Sview(v_{max},t))$, i.e., the center graph of $\alpha$'s snapshot view from node $v_{max}$. In the first case, Lemma \ref{lemma-1} tells us that the snapshot view from node $v_{max}$ contains all the nodes of $G$. In the second case, we observe that $v_{closest}$ is in the center of $G$ since it is in the center graph of $\alpha$'s snapshot view from node $v_{max}$ (which contains all nodes of $G$). But by the definition of center, the distance from $v_{closest}$ to any node in $G$ is at most $R \leq H$, so all nodes of $G$ are in the snapshot view from $v_{closest}$ as well. \end{proof} 

\begin{lemma} \label{lemma-3}
	Suppose that every non-faulty robot starts an execution of March-to-Center in the same round $t' > 0$. In round $t'+H$, every non-faulty robot has the same snapshot view.
\end{lemma}

\begin{proof}
	We see from the description of March-to-Center that there can be two cases in each execution: moving along the path $\pi$ for $|\pi|$ rounds followed by a waiting period of length $H-|\pi|$, or, a waiting period of length $H$. In both cases, the execution takes exactly $H$ rounds. Now, by Lemma $\ref{lemma-2}$, we see that at the end of the execution, i.e., in round $t'+H$, each robot's snapshot view is the entire graph.  \end{proof}

\begin{lemma} \label{lemma-4}
	For any positive integers $i$ and $t'$, suppose that every non-faulty robot starts its $i^{th}$ execution of March-to-Center in round $t'$. Then, at the start of round $t'+H$,  exactly one of the following is true: (i) every non-faulty robot sets $v_{target}$ equal to a non-null value, or, (ii) every non-faulty robot has $v_{target}$ equal to null, and they all start their $(i+1)^{th}$ execution of March-to-Center.
\end{lemma}
\begin{proof}
	By Lemma \ref{lemma-3}, in round $t'+H$, every robot gets the same snapshot view $S$. There are two cases to consider. In the first case, suppose that there is a singleton ID in the center graph of $S$. Then, according to the description of March-to-Center, every non-faulty robot sets its variable $v_{target}$ to the node that contains a robot with the smallest singleton ID, which implies that every non-faulty robot has $v_{target}$ equal to a non-null value. In the second case, suppose that there is no singleton ID in the center graph of $S$. Then, according to the description of March-to-Center, $v_{target}$ at each non-faulty robot remains null. According to the description of the H-View-Algorithm, this means that all non-faulty robots will execute March-to-Center again.  \end{proof}

We now proceed to show that each execution of March-to-Center by the non-faulty robots is started at the same time. This is useful because it means that the robots make decisions using the same snapshot view, which minimizes the influence of the Byzantine robots: if a Byzantine robot imitates a non-faulty robot's ID $l$ in a fixed round $t$, then it cannot imitate any other ID's in the same round.

\begin{lemma}\label{lemma-5}
	For any positive integer $k$, suppose that all non-faulty robots start their $k^{th}$ execution of March-to-Center and have $v_{target}=null$. For every positive integer $i \leq k$, every non-faulty robot starts executing its $i^{th}$ execution of March-to-Center in round $(m+2)n^2 + (i-1)H$.
\end{lemma}

\begin{proof}
	We prove the statement by induction on $i$.
	
	\textbf{Base case:} From the description of the H-View-Algorithm, each non-faulty robot executes March-to-Center for the first time starting in round $(m_{max}+2)\cdot n^2_{max} = (m+2)\cdot n^2$. Thus, the statement is true for $i=1$.
	
	\textbf{Inductive step:} Assume that, for some $j \in \{1,\ldots,k-1\}$, the statement is true for $i=j$. In particular, assume that every robot started its $j^{th}$ execution of March-to-Center in round $(m+2)n^2 + (j-1)H$. By the description of March-to-Center, there can be two cases in their $j^{th}$ execution: moving along the path $\pi$ for $|\pi|$ rounds followed by a waiting period of length $H-|\pi|$, or, a waiting period of length $H$. In both cases, the execution takes exactly $H$ rounds. By Lemma \ref{lemma-4} and the fact that no robot has set its $v_{target}$ variable to a non-null value before the $k^{th}$ execution, we get that in round $(m+2)n^2 + (j-1)H+H=(m+2)n^2 + j \cdot H$, every robot starts its $(j+1)^{th}$ execution of March-to-Center.  \end{proof}

\begin{lemma} \label{lemma-6}
	Let $k > 0$ be the smallest integer such that at least one non-faulty robot sets its $v_{target}$ to a non-null value during its $k^{th}$ execution of March-to-Center, and suppose that this execution of March-to-Center starts in round $t'$. Then, every non-faulty robot sets $v_{target}$ to the same value at the start of round $t'+H$.
\end{lemma}

\begin{proof}
	By Lemma \ref{lemma-5}, for every positive integer $i\leq k$, every robot starts its $i^{th}$ execution of March-to-Center in the same round, so all robots start the $k^{th}$ execution of March-to-Center in round $t'$. Lemma \ref{lemma-4} implies that, at the start of round $t'+H$, either every robot sets a non-null value of $v_{target}$, or, variable $v_{target}$ is null for every robot. The second case does not occur since we know that at least one non-faulty robot sets its $v_{target}$ to a non-null value during its $k^{th}$ execution of March-to-Center. Therefore, the first case occurs: all robots set their $v_{target}$ to a non-null value at the start of round $t'+H$. Moreover, by Lemma \ref{lemma-3}, every robot has the same snapshot view $S$ in round $t'+H$. Hence, by the description of March-to-Center, every robot sets its $v_{target}$ to the same node: the node that contains a robot with smallest singleton ID in the center graph of $S$.  \end{proof}

\begin{corollary}\label{col-1}
	If there exists a positive integer $k$ such that at least one non-faulty robot sets its $v_{target}$ to a non-null value during its $k^{th}$ execution of March-to-Center, then all non-faulty robots set $v_{target}$ to the same non-null value at the start of round $(m+2)n^2 + kH$.
\end{corollary}

% \begin{lemma} \label{lemma-7}
% Every robot starts every execution of March-to-Center at the same time.
% \end{lemma}
% \begin{proof}
% In view of Lemma \ref{lemma-4}, before getting a non-null value of $v_{target}$ by any non-faulty robot, every robot starts each execution of March-to-Center at the same time. But, from Lemma \ref{lemma-4}, it follows that every robot actually sets its variable $v_{target}$ at the same time. Hence, every robot starts every execution of March-to-Center at the same time.  \end{proof}

We now set out to show that all robots set their $v_{target}$ variable to a non-null value within $f+1$ executions of March-to-Center. The idea behind the proof is to show that, in each execution of March-to-Center that ends with $v_{target} = null$, at least one non-faulty robot makes progress towards determining its correct location within its snapshot view. Once there are enough robots that have determined their correct location (more than the number of Byzantine robots), we are guaranteed to have at least one singleton ID appear in the center of the graph, and all robots will set their $v_{target}$ as the location of the smallest such ID. 

To formalize the argument, we introduce a function $\Phi$ that measures how much progress has been made by all robots towards determining their correct location within their snapshot view. In what follows, for each $t \geq (m+2)\cdot n^2$, we denote by $P_{\alpha,t}$ the value of variable $P_\alpha$ at robot $\alpha$ in round $t$. From the description of the H-View-Algorithm, recall that $P_\alpha$ is set by each robot $\alpha$ for the first time in round $(m+2)\cdot n^2$, and the value assigned in this round is the set of nodes in $\alpha$'s snapshot view that match its local view, i.e., the nodes that have the same degree and the same list of robot ID's as $\alpha$'s current location. In subsequent rounds, the only changes to $P_\alpha$ involve the removal of nodes, so $P_{\alpha,t+1} \subseteq P_{\alpha,t}$ for all $t > (m+2)\cdot n^2$. For any fixed round $t \geq (m+2)\cdot n^2$, we denote by $\Phi_t$ the sum $\sum_\alpha |P_{\alpha,t}|$, which is taken over all non-faulty robots $\alpha$. We now prove some useful bounds on $\Phi_t$ and how its value changes in each execution of March-To-Center.
%The following propositions follow from the fact that each of the $f$ Byzantine robots has exactly one ID in any round, and there are at least $f+1$ non-faulty robots.

% \begin{proposition} \label{pro-1}
% In any round $t'$, there must be at least one non-faulty robot $\alpha$ whose ID will be a singleton ID in $G$.
% \end{proposition}

\begin{proposition} \label{pro-1}
	In any round $t \geq (m+2)\cdot n^2$, we have $m-f \leq \Phi_t \leq m$.
	%the summation of $|{P_{\alpha}}|$ of every robot is at least $m-f$ and at most $m$.
	%\begin{equation} \label{eq1}
	%{Sum}_{|P_{\alpha}|}=\sum\limits_{\textit{ $\alpha$}}|{P_{\alpha}}| \leq m 
	%\quad \text{and}\quad
	%{Sum}_{|P_{\alpha}|} \geq m-f 
	%\end{equation}
\end{proposition}
\begin{proof}
	First, we show that $\Phi_t \leq m$. Since each $P_{\alpha,t}$ only contains nodes where the ID $l_\alpha$ appears in round $t$, it follows that $|P_{\alpha,t}|$ is bounded above by the number of robots whose ID in round $t$ is equal $l_\alpha$. As each robot has exactly one ID in round $t$ (including the Byzantine robots), it follows that $\Phi_t = \sum_\alpha |P_{\alpha,t}| \leq m$.
	Next, to show that $\Phi_t \geq m-f$, we observe that there are $m-f$ non-faulty robots, and each non-faulty robot $\alpha$ has $|P_{\alpha,t}| \geq 1$ in every round $t \geq (m+2)\cdot n^2$. This is because a non-faulty robot $\alpha$ only removes a node $v$ from $P_\alpha$ if it performs March-to-Center or Merge under the assumption that it starts the execution from node $v$ in its snapshot view, but notices an inconsistency between this assumption and its observed experience. Since $\alpha$'s actual $v_{max}$ node from which it starts March-to-Center or Merge would not result in any inconsistency, this node would never be removed from $P_\alpha$, which implies that $|P_{\alpha}| \geq 1$ after the first round in which $P_\alpha$ is given a value.
\end{proof}

\begin{lemma} \label{lemma-7}
	Consider any execution of March-to-Center by the non-faulty nodes, and suppose that the execution starts in round $t'$. Then, exactly one of the following occurs: (i) all non-faulty robots set their $v_{target}$ variable to a non-null value at the start of round $t' + H$, or, (ii) $\Phi_{t'+H} \leq \Phi_{t'} - 1$.
		% ${Sum}_{|P_{\alpha}|}$ goes down by $1$ for at least one non-faulty robot $\alpha$.
	
\end{lemma}
\begin{proof}
	By Lemma \ref{lemma-4}, exactly one of the following occurs at the start of round $t'+H$:
	\begin{itemize}
		\item All non-faulty robots set their $v_{target}$ variable to some non-null value, or,
		
		\item Variable $v_{target}$ is null for every robot. By Lemma \ref{lemma-3}, we know that in round $t+H'$, all non-faulty robots have the same snapshot view $S$, and, by Lemma \ref{lemma-2}, $S$ contains all the nodes of $G$. As there are at least $f+1$ non-faulty robots and exactly $f$ Byzantine robots, there must be at least one non-faulty robot $\beta$ whose ID will be a singleton ID in $S$. But since $v_{target}$ is null for every non-faulty robot, this implies that there is no singleton ID in $C(S)$ in round $t'+H$, and so $\beta$ is located outside of $C(S)$. According to the description of March-to-Center, it must be the case that $|P_\beta|>1$ in round $t'$, because otherwise $\beta$ would have moved to a node in the center of its snapshot view in this execution of March-to-Center. Consequently, according to March-to-Center, the robot $\beta$ removes all other nodes from $P_{\beta}$ except the one node that contains its ID $l_\beta$ (as $l_\beta$ is a singleton ID). Thus, the value of $|P_{\beta}|$ decreases during some round in the range $t',\ldots,t'+H$, so it follows that $\Phi_{t'+H} \leq \Phi_{t'} - 1$. \end{itemize}  \end{proof}

\begin{theorem} \label{theo-1}
	There exists a positive integer $k \leq f+1$ such that every non-faulty robot sets its variable $v_{target}$ to the same non-null value at the start of round $(m+2)n^2+kH$.
\end{theorem} 

\begin{proof}
	First, suppose that there is at least one non-faulty robot that sets its $v_{target}$ to a non-null value during one of its first $f$ executions of March-to-Center. In this case, the desired result follows directly from Corollary \ref{col-1}. So, in what follows, we assume that all non-faulty robots have $v_{target}=null$ during the first $f$ executions of March-to-Center. Therefore, all non-faulty robots start their $(f+1)^{th}$ execution of March-to-Center with $v_{target}=null$, and by Lemma \ref{lemma-5}, they start this execution in round $(m+2)n^2 + fH$. By Lemmas \ref{lemma-2} and \ref{lemma-3}, each non-faulty robot starts this execution with the same snapshot view, which we'll denote by $S$, that contains all the nodes of $G$.
	
	By Lemma \ref{lemma-7}, after each of the first $f$ executions of March-to-Center, the value of $\Phi$ decreases by at least 1. It follows that $\Phi_{(m+2)n^2 + fH} \leq \Phi_{(m+2)n^2} - f$. However, by Proposition \ref{pro-1}, we know that $\Phi_{(m+2)n^2} \leq m$ and $\Phi_{(m+2)n^2 + fH} \geq m-f$, so altogether we conclude that $\Phi_{(m+2)n^2 + fH} = m-f$. But $m-f$ is the number of non-faulty robots, so the sum $\Phi_{(m+2)n^2 + fH} = \sum_{\alpha}|P_{\alpha,(m+2)n^2 + fH}|$ has $m-f$ non-zero terms. This implies that each $|P_{\alpha,(m+2)n^2 + fH}|$ is equal to exactly 1. Therefore, by the description of March-to-Center, all non-faulty robots move to a node in the center graph of their snapshot view $S$. This means that there are at least $f+1$ non-faulty robots in the center of $S$ in round $(m+2)n^2 + (f+1)H$, and at least one of their ID's is a singleton ID since there are at most $f$ Byzantine nodes. Thus, by the description of March-to-Center, every non-faulty robot sets its $v_{target}$ to the same node: the node that contains a robot with smallest singleton ID in the center graph of $S$, which proves the desired statement with $k=f+1$. \end{proof}

Now we come to the third part of the algorithm which consists of the executions of Merge. By the description of the H-View-Algorithm, non-faulty robots start executing their Merge operation immediately after setting a non-null value of $v_{target}$. Moreover, by Theorem \ref{theo-1}, we see that every robot sets its $v_{target}$ variable to the same non-null value in the same round, and so every non-faulty robot starts executing its first execution of Merge at the same time as well. More specifically, we denote by $k$ the number of executions of March-to-Center performed by the non-faulty robots, and conclude that all non-faulty robots start their first execution of Merge in round $(m+2)n^2 + kH$. By the description of Merge, each execution of Merge consists of exactly $H$ rounds, and according to the H-View-Algorithm, an additional $H$ rounds are then used to perform the steps of Merge in reverse. These observations imply the following fact.
% % Here, according to our algorithm, $v_0$ and $v_{target}$ both lie on its current snapshot, meaning the distance between $v_0$ and $v_{target}$, $d(v_0$, $v_{target})\leq H$. 
\begin{lemma} \label{lemma-9}
	For any positive integer $i$, if an $i^{th}$ execution of Merge is performed, then all non-faulty robots start this execution in round $(m+2)n^2 + (k+2(i-1))H$.
\end{lemma}

%In each execution of Merge, every non-faulty robot guesses one node from its list of possible positions $P_\alpha$ (in the snapshot got from the end of the last execution of March-to-Center), computes a sequence of port numbers that leads to $v_{target}$, and tries to follow it in the real graph. 
Our final goal is to show that all non-faulty robots gather at $v_{target}$ after at most $(f+2)-k$ executions of Merge, where $k$ is the number of March-to-Center operations executed by the non-faulty robots. Before proving this in Theorem \ref{theo-2}, we establish the following technical results.

\begin{lemma} \label{lemma-10}
	For any $t \geq 0$, suppose that $v$ is a node such that at least $m-f$ robots are located at $v$ at the start of round $t$. Then, the local view at $v$ in round $t$ is unique. More precisely, for any node $v' \neq v$, we have $Lview(v',t) \neq Lview(v,t)$.
	%Let, for some $t>0$, $v$ be a node such that there were at least $m-f$ robots at $v$ in round $t$. At round $t$, the local view $Lview(v,t)$ must be unique in $G$. More precisely, for any node $v'$, $(v' \neq v) \leftrightarrow (Lview(v',t) \neq Lview(v,t))$, and $(v' = v) \leftrightarrow (Lview(v',t) = Lview(v,t))$.
\end{lemma}

\begin{proof}
	For any $v,v'$ such that $v \neq v'$, if there are at least $m-f$ robots at $v$ in round $t$, there can be at most $f$ robots at $v'$ in round $t$. Since there are at least $f+1$ non-faulty robots, it follows that $m \geq 2f+1$, so $m-f > f$. In particular, this means that the number of ID's in $Lview(v,t)$ is strictly greater than the number of ID's in $Lview(v',t)$, so $Lview(v,t) \neq Lview(v',t)$. 
	%Hence for any node $v'$, $v' \neq v \rightarrow Lview(v',t) \neq Lview(v,t)$, and $Lview(v',t) = Lview(v,t) \rightarrow v' = v$. Bidirectionally, it is straight forward that if, two snapshot views are not same in the same round, then it is not possible that they are captured at the same node. Hence, $Lview(v',t) \neq Lview(v,t) \rightarrow v' \neq v$. Similarly, $v'=v \rightarrow Lview(v',t) = Lview(v,t)$, since $v'$ and $v$ are the same node.
	\end{proof}

% First we will show that each execution of Merge finishes in one of two ways: all robots have gathered, or, at least one robot has eliminated one incorrect guess about its starting position.

\begin{lemma} \label{lemma-11}
	Consider any execution of Merge by the non-faulty nodes, and suppose that the execution starts in round $t'$. Then at least one of the following holds: (i) all non-faulty robots are gathered at $v_{target}$ in round $t'+H$, or, (ii) $\Phi_{t'+H} \leq \Phi_{t'} - 1$.	
\end{lemma}

\begin{proof}
	Assume that (i) does not hold in round $t'+H$, i.e., at least one non-faulty robot is not located at $v_{target}$ in round $t'+H$. There are two possibilities: 
	\begin{itemize}
		\item {\bf There are at least $m-f$ robots at $v_{target}$ in round $t'+H$.} By Lemma $\ref{lemma-10}$, each robot $\beta$ that is at a node $v' \neq v_{target}$ in round $t'+H$ has a local view $Lview(v',t'+H)$ that is different than $Lview(v_{target}, t'+H)$. Hence, according to the description of Merge, each such robot $\beta$ removes a node from its $P_\beta$, i.e., the value of $|P_\beta|$ decreases in some round in the range $t',\ldots,t'+H$. It follows that $\Phi_{t'+H} \leq \Phi_{t'}-1$.
		\item {\bf There are fewer than $m-f$ robots at $v_{target}$ in round $t'+H$.} As the number of non-faulty robots is $m-f$, it follows that there is at least one non-faulty robot $\alpha$ whose ID $l_\alpha$ is not seen at $v_{target}$ in $\alpha$'s snapshot view in round $t'+H$. Hence, according to the description of Merge, $\alpha$ removes a node from its $P_\alpha$, i.e., the value of $|P_\alpha|$ decreases in some round in the range $t',\ldots,t'+H$. It follows that $\Phi_{t'+H} \leq \Phi_{t'}-1$.
	\end{itemize}
	
	% in round $t'+H$, there is exactly one of the following can occur.
	% \begin{itemize}
	%     \item Every non-faulty robot has successfully reached the target node $v_{target}$. Therefore, there must be at least $m-f$ robots at $v_{target}$, and so the local view $Lview(v_{target},t'+H)$ will be unique in $G$ by Lemma \ref{lemma-10}. On the other hand, at $v_{target}$, every robot gets the same snapshot view $Sview(v_{target},t'+H)$ that contains all the nodes of $G$ (since $v_{target}$ is a center of the graph). In $Sview(v_{target},t'+H)$, the local view of each robot exactly matches with $Lview(v_{target},t'+H)$, and there is no other node $v'$ such that $Lview(v', t'+H)=Lview(v_{target},t'+H)$ in view of Lemma \ref{lemma-10}. Hence, according to the algorithm, no non-faulty robot $\alpha$ removes its current guessed node $v_0$ from $P_{\alpha}$, ${Sum}_{|P_{\alpha}|}$ remains the same.               
	
	%\item Some non-faulty robots are not located at $v_{target}$ in round $t'+H$.

	% \end{itemize}

\end{proof}

\begin{lemma}\label{atTarget}
	During the execution of the H-View-Algorithm, if $k \geq 1$ executions of March-to-Center are performed followed by $f+2-k$ executions of Merge, then all non-faulty robots are gathered at $v_{target}$.
	%Let, for some $k'>0$, $k'$ be the number of total executed March-to-Center by the non-faulty robots. In $((f+2)-k')$ executions of Merge, all non-faulty robots gather at $v_{target}$.
\end{lemma}

\begin{proof}
	By the description of the H-View-Algorithm and Corollary 1, if $k$ executions of March-to-Center are performed, then $v_{target}$ was set for the first time by all non-faulty robots at the end of the $k^{th}$ execution of March-to-Center. By Lemma \ref{lemma-7}, after each of the first $k-1$ executions of March-to-Center, the value of $\Phi$ decreases by at least 1. It follows that $\Phi_{(m+2)n^2 + (k-1)H} \leq \Phi_{(m+2)n^2} - (k-1)$. By Proposition \ref{pro-1}, we know that $\Phi_{(m+2)n^2} \leq m$, so it follows that $\Phi_{(m+2)n^2 + (k-1)H} \leq m-(k-1)$. Since the value of $\Phi$ never increases (the algorithm only ever removes nodes from the $P_\alpha$ sets) it follows that $\Phi_{(m+2)n^2 + kH} \leq m-(k-1)$ as well, where round $(m+2)n^2 + kH$ is when the first Merge execution begins. Now, we consider the first $f+1-k$ executions of Merge by the non-faulty robots, and we consider two cases:
	
	\begin{itemize}
		\item Suppose that, for some $i \in \{1,\ldots,f+1-k\}$, all non-faulty robots are gathered at $v_{target}$ at the end of the $i^{th}$ execution of Merge. Since the number of non-faulty robots is $m-f$, it follows that there would be at least $m-f$ robots at $v_{target}$. By Lemma \ref{lemma-10}, the local view at $v_{target}$ would be unique in $G$, and the local view of each non-faulty robot would exactly match it. Hence, according to the description of Merge, no non-faulty robot would modify its $P_\alpha$ set, and so the next execution of Merge (if any) would start from the same node $v_0$. It follows that in all subsequent executions of Merge (in particular, the $(f+2-k)^{th}$ execution) all non-faulty robots will be gathered at $v_{target}$.
		%Let, for some $i>0$, $i$ be an execution in $\textit{I}$ such that all non-faulty robots have gathered at $v_{target}$ in the $i^{th}$ execution of Merge. At the end of the execution, let's say in round $t'$, there would be at least $m-f$ robots at $v_{target}$. Now, in view of Lemma \ref{lemma-10}, the local view $Lview(v_{target},t')$ would be unique in $G$, and the local view of each robot would exactly match with $Lview(v_{target},t')$. Hence, according to the algorithm, no non-faulty robot $\alpha$ would remove its current guessed node $v_0$ from $P_{\alpha}$. Consequently, in execution $i+1$, each robot would follow the same port sequence as in execution $i$, and would be gathered at $v_{target}$ again. It continues, and robots gather in all subsequent executions.     
		
		\item Suppose that, for every $i \in \{1,\ldots,f+1-k\}$, at least one non-faulty robot is not located at $v_{target}$ at the end of the $i^{th}$ execution of Merge. Then, according to Lemma \ref{lemma-11}, the value of $\Phi$ decreases by at least 1 in each such execution. As the value of $\Phi$ was bounded above by $m-(k-1)$ at the start of the first Merge execution, and it decreases by at least $f+1-k$ during the first $f+1-k$ executions of Merge, it follows that, after the $(f+1-k)^{th}$ execution of Merge, the value of $\Phi$ is at most $m-f$. However, by Proposition \ref{pro-1}, we know that $\Phi$ is at least $m-f$, so altogether we conclude that the value of $\Phi$ after the $(f+1-k)^{th}$ execution of Merge is exactly $m-f$. But $m-f$ is the number of non-faulty robots, so the summation represented by $\Phi$ has $m-f$ non-zero terms. This implies that each $|P_{\alpha}|$ is equal to exactly 1 for each non-faulty robot $\alpha$. Then, in the final execution of Merge, i.e., in execution $f+2-k$, each non-faulty robot will compute a path to $v_{target}$ using its snapshot view, but using its actual location as starting node $v_0$. This means that all non-faulty nodes will be located at $v_{target}$ at the end of execution $f+2-k$ of Merge.
	\end{itemize}
\end{proof}

Finally, we verify that the H-View-Algorithm ensures that Merge is executed at least $f+2-k$ times after $k$ executions of March-to-Center. The Merge operation is executed until the value of $phase$ is greater than $\lceil m/2 \rceil$, and from the assumption that the number of non-faulty robots is at least $f+1$, we know that $m \geq 2f+1$. In particular, this means that the combined number of March-to-Center and Merge executions is at least $f+1$, and then one more Merge is executed after exiting the `repeat' loop. This concludes the proof of correctness of the H-View-Algorithm.

\begin{theorem} \label{theo-2}
	In any $n$-node graph with radius $R$, if the H-View-Algorithm is performed by any team of $m$ robots consisting of $f$ Byzantine robots and at least $f+1$ non-faulty robots with visibility $H \geq R$, then Gathering is solved within $(m+2)\cdot n^2+ H\cdot m \in O(mn^2)$ rounds.
\end{theorem}

\begin{proof}
	By Lemma $\ref{lemma-1}$, every non-faulty robot spends exactly $(m+2)n^2$ rounds for the \textbf{Find-Lookout} operation. Then, by Theorem $\ref{theo-1}$, there exists a positive integer $k\leq f+1$ such that every non-faulty robot sets its variable $v_{target}$ at the start of the round $(m+2)n^2+kH$. More precisely, robots spend exactly $kH$ rounds performing the $\textbf{March-to-Center}$ executions. After that, every robot spends exactly $(\lceil m/2 \rceil-k)2H+H$ rounds for its \textbf{Merge} executions, after which all non-faulty are located at $v_{target}$ (by Lemma \ref{atTarget}. In total, the number of rounds is $(m+2)\cdot n^2 + kH+(\lceil m/2 \rceil-k)2H+H$. For the minimum value of $k=1$, we get that the robots use at most $(m+2)\cdot n^2+ H\cdot m$ rounds to accomplish the gathering. As $H \leq n$ (at most full visibility), the number of rounds is in $O(mn^2)$, i.e., polynomial in the network size and team size.
\end{proof}

\section{Impossibility Results} \label{impossibilty}

%\section{The Lower Bound}   \label{lowerbound}  
%In this section, we show that, under the model assumptions specified in Section \ref{Model}, our algorithm is optimal with regards to two assumptions: the number of non-faulty robots in the team, and the amount of visibility. 
%In our proposed algorithm, we assumed the number of non-faulty robot should be at least $f+1$ (in the presence of $f$ Byzantine robots). In view of Theorem 4.7 \cite{dieudonne2014gathering}, we see that if the number of non-faulty robots is less than $f+1$, then this is impossible for any deterministic algorithm to accomplish the gathering. 

First, we recall Theorem 4.7 from \cite{dieudonne2014gathering}, which states that there is no deterministic algorithm that solves Gathering in the presence of $f$ Byzantine robots if the number of non-faulty agents is at most $f$ (and these non-faulty agents know the size of the graph). This impossibility result was proven in a model where robots have no visibility beyond their local view (i.e., visibility $H = 0$). However, the same proof works under the assumption that each non-faulty robot has full visibility of the entire graph in every round, which proves that our algorithm is optimal with respect to the number of non-faulty robots in the team.

\begin{theorem} 
	There is no deterministic algorithm that solves Gathering if the number of Byzantine robots in the team is $f$ and the number of non-faulty robots is at most $f$, even if the non-faulty agents have visibility $H$ equal to the diameter of the graph.
\end{theorem}
% \begin{proof}
% In order to prove the Theorem 4.7\cite{dieudonne2014gathering}, the authors showed two different executions - ${EX}_1$ and ${EX}_2$ in an oriented ring graph $G$ (there were $f$ non-faulty and $f$ Byzantine robots in $G$). In view of the proof, we see that ${EX}_1$ and ${EX}_2$ were identical to some non-faulty robot $\alpha$. Consequently, $\alpha$ made an early termination in $EX_2$ before meeting other non-faulty robots. Though the theorem was proved in terms of the model where the robots have the local view only (i.e., visibility range $H=0$), it works for our model as well. Necessarily, ${EX}_1$ and ${EX}_2$ would be identical to $\alpha$, even if the robot gets a snapshot of the full graph in each round. Hence we get the following Theorem.
% \end{proof}

%We see our proposed algorithm matches the lower bound of $f+1$ for the required number of non-faulty robots. In our algorithm, we also considered the visibility range of each non-faulty robot, $H \geq R$. 
Next, we prove that to solve Gathering in arbitrary graphs, the visibility $H$ of each non-faulty robot must somehow depend on the radius of the graph. In particular, it is not sufficient to fix some constant visibility range. We remark that this does not contradict the existence of previously-known algorithms that work when $H=0$, as those algorithms make additional assumptions that are not present in our model (e.g., knowledge of the graph size, knowledge of the number of Byzantine robots, or whiteboards at the nodes).
%As the H-View-Algorithm works when $H \geq R$, this proves that our algorithm is optimal with regards to the amount of visibility possessed by the non-faulty robots.

\begin{theorem}\label{HLB}
	There is no deterministic algorithm that can solve Gathering when executed in any graph by any team of $m$ robots consisting of $f \geq 0$ Byzantine robots and at least $f+1$ non-faulty robots if the visibility range $H$ of each non-faulty robot is a fixed constant $c$.
\end{theorem}

\begin{proof}
	Let $c$ be any fixed positive integer. To obtain a contradiction, assume the existence of a deterministic algorithm $A$ that can solve Gathering when executed in any graph by any team of $m$ robots consisting of $f \geq 0$ Byzantine robots and at least $f+1$ non-faulty robots if the visibility range $H$ of each non-faulty robot is equal to $c$.
	
	First, we construct an instance consisting of a cycle graph $C_1=(V_1,E_1)$ with an even number of nodes $|V_1| = 2c+2$. The radius $R_1$ of $C_1$ is $c+1$. At each node $v \in V_1$, the two incident edges are labeled with port numbers $0$ and $1$ such that $0$ leads clockwise and $1$ leads anticlockwise. The initial positions of the robots in $C_1$ are as follows: a non-faulty robot $\alpha$ with ID $l_\alpha$ is placed at some node $v_0$, and a non-faulty robot $\beta$ with ID $l_\beta$ at a node $w$ such that the distance $d(v_0,w)=R_1=c+1$. There are no Byzantine robots in $C_1$. Consider the execution ${EX}_1$ of algorithm $A$ on instance $C_1$. As $A$ is assumed to be a correct algorithm, there exists some round $r_1$ in which robots $\alpha$ and $\beta$ have terminated and gathered at some node $v_{target} \in V_1$. 
	%We denote by $t =1,2, \cdots, i$ the sequence of consecutive rounds in ${EX}_1$ from starting of the execution till end. 

	Next, we construct a second instance consisting of a cycle graph $C_2 = (V_2,E_2)$ with an even number of nodes $|V_2| = 4r_1+2(c+1)$. The radius of $R_2$ of $C_2$ is $2r_1+c+1$. At each node $v \in V_2$, the two incident edges are labeled with port numbers $0$ and $1$ such that $0$ leads clockwise and $1$ leads anticlockwise. The initial positions of the robots in $C_2$ are as follows: the non-faulty robot $\alpha$ with ID $l_\alpha$ is placed at node $v_0$ (as in the first instance $C_1$ above), a Byzantine robot with ID $l_\beta$ is placed at a node $v_{CW}$ that is distance exactly $c+1$ away from $v_0$ in the clockwise direction, and another Byzantine robot with ID $l_\beta$ is placed at a node $v_{ACW}$ that is distance exactly $c+1$ away from $v_0$ in the anticlockwise direction. Further, we place $2$ non-faulty robots at a node $w$ such that $d(v_0,w) = R_2 = 2r_1 + c + 1$. These 2 non-faulty robots have distinct ID's that are not equal to $l_\alpha$ or $l_\beta$. The number of Byzantine robots is $f=2$, and there are $3 = f+1$ non-faulty robots (one at $v_0$ and two at $w$). We denote by ${EX}_2$ the execution of algorithm $A$ on instance $C_2$.
	
	We now demonstrate that the Byzantine robots in $C_2$ can behave in such a way that, for each round $t$, the robot $\alpha$ with ID $l_\alpha$ cannot distinguish between executions ${EX}_1$ and ${EX}_2$, i.e., robot $\alpha$'s local view and snapshot view in every round are the same across both executions. This leads to a contradiction: since $\alpha$ terminates its algorithm in round $r_1$ in execution ${EX}_1$, it will also terminate its algorithm in round $r_1$ in execution ${EX}_2$, and since the initial distance between $\alpha$ and the other non-faulty robots is strictly greater than $2r_1$, it follows that $\alpha$ terminates before the non-faulty robots can gather.
	
	First, note that $\alpha$'s visibility range is $c$ in both executions, which means that its snapshot view consists of $2c+1$ nodes in every round of both executions. By the initial placement of the robots in both executions, we note that in round $t=0$ of both executions, there are no robots within distance $c$ of $\alpha$'s initial position $v_0$. So, $\alpha$'s local view in round 0 of both executions consists of a node of degree 2 containing the ID $l_\alpha$, and, $\alpha$'s snapshot view in round 0 of both executions consists of a path of length $2c+1$ nodes with only ID $l_\alpha$ located at the middle node. Further, we note that the two other non-faulty robots in $C_2$ are never visible to $\alpha$ in execution ${EX}_2$: their initial distance to $\alpha$ is $2r_1+c+1$, so in round $r_1$, each of their distances to $\alpha$ is at least $c+1$.
	
	To define the behaviour of the Byzantine robots in $C_2$ during rounds $t = 1,\ldots,r_1$ of execution ${EX}_2$, we observe the execution ${EX}_1$. In particular:
	\begin{itemize}
		\item For each round $t > 0$ of ${EX}_1$ in which $\alpha$ does not see $\beta$ in its snapshot view: the Byzantine robots follow the same port in round $t-1$ of ${EX}_2$ as $\alpha$ did in round $t-1$ in ${EX}_1$. Doing so ensures that both Byzantine robots remain at distance $c+1$ from $\alpha$ at the start of round $t$ in ${EX}_2$, i.e., are not in $\alpha$'s snapshot view.
		\item For each round $t > 0$ of ${EX}_1$ in which $\alpha$ sees $\beta$ in its snapshot view but did not see $\beta$ in its snapshot view in round $t-1$: the Byzantine robot on the appropriate side of $\alpha$ (clockwise or counterclockwise) moves so that it appears at the same node in $\alpha$'s snapshot view in round $t$ of ${EX}_2$ as $\beta$ does in round $t$ of ${EX}_1$. The other Byzantine robot follows the same port as $\alpha$ does in round $t-1$ (so that its distance from $\alpha$ at the start of round $t$ is still $c+1$, i.e., it does not appear in $\alpha$'s snapshot view).
		\item For each round $t > 0$ of ${EX}_1$ in which $\alpha$ sees $\beta$ in its snapshot view and also saw $\beta$ in its snapshot view in round $t-1$: the Byzantine robot that was in $\alpha$'s snapshot view in round $t-1$ of ${EX}_2$ follows the same port in round $t-1$ of ${EX}_2$ as $\beta$ did in round $t-1$ of ${EX}_1$. The other Byzantine robot follows the same port as $\alpha$ does in round $t-1$ (so that its distance from $\alpha$ at the start of round $t$ is still $c+1$, i.e., it does not appear in $\alpha$'s snapshot view).
	\end{itemize}
	
	It is clear from this behaviour that $\alpha$ sees the same thing up to round $r_1$ in both executions ${EX}_1$ and ${EX}_2$: when $\alpha$ sees no other robots in round $t$ of ${EX}_1$, then both Byzantine robots move so that they are both at distance $c+1$ from $\alpha$ in round $t$ of ${EX}_2$; moreover, when $\alpha$ sees $\beta$ in round $t$ of ${EX}_1$, then one Byzantine robot (which has ID $l_\beta$) moves so that its position relative to $\alpha$ in round $t$ of ${EX}_2$ is the same as $\beta$'s relative position to $\alpha$ in round $t$ of ${EX}_1$, while the other Byzantine robot moves so that it is at distance $c+1$ from $\alpha$ in round $t$ of ${EX}_2$.

\end{proof}

We were not able to extend the lower bound argument in Theorem \ref{HLB} to a non-constant visibility range $H$. The reason is that, when we change the underlying graph, the visibility radius of a robot is different in the new graph, so we cannot use indistinguishability to conclude that a robot will behave in the same way in both graphs. Establishing a lower bound on $H$ with respect to $R$ is left as an open problem.

%\subsubsection*{Acknowledgements.} The authors acknowledge the support of the Natural Sciences and Engineering Research Council of Canada (NSERC), Discovery Grant RGPIN--2017--05936.

\bibliographystyle{plain}
\bibliography{bibliography}

\begin{thebibliography}{10}

\bibitem{Barra}
Eduardo~Mesa Barrameda, Nicola Santoro, Wei Shi, and Najmeh Taleb.
\newblock Sensor deployment by a robot in an unknown orthogonal region:
  Achieving full coverage.
\newblock In {\em 20th {IEEE} International Conference on Parallel and
  Distributed Systems, {ICPADS} 2014}, pages 951--960, 2014.

\bibitem{Barri}
Lali Barri{\`{e}}re, Paola Flocchini, Eduardo~Mesa Barrameda, and Nicola
  Santoro.
\newblock Uniform scattering of autonomous mobile robots in a grid.
\newblock {\em Int. J. Found. Comput. Sci.}, 22(3):679--697, 2011.

\bibitem{BhagatSurvey}
Subhash Bhagat, Krishnendu Mukhopadhyaya, and Srabani Mukhopadhyaya.
\newblock Computation under restricted visibility.
\newblock In {\em Distributed Computing by Mobile Entities, Current Research in
  Moving and Computing}, pages 134--183. Springer, 2019.

\bibitem{bouchard2016byzantine}
S{\'e}bastien Bouchard, Yoann Dieudonn{\'e}, and Bertrand Ducourthial.
\newblock Byzantine gathering in networks.
\newblock {\em Distributed Computing}, 29(6):435--457, 2016.

\bibitem{bouchard2018byzantineICALP}
S{\'{e}}bastien Bouchard, Yoann Dieudonn{\'{e}}, and Anissa Lamani.
\newblock Byzantine gathering in polynomial time.
\newblock In {\em 45th International Colloquium on Automata, Languages, and
  Programming, {ICALP} 2018}, pages 147:1--147:15, 2018.

\bibitem{ChalFaults}
J{\'{e}}r{\'{e}}mie Chalopin, Yoann Dieudonn{\'{e}}, Arnaud Labourel, and
  Andrzej Pelc.
\newblock Rendezvous in networks in spite of delay faults.
\newblock {\em Distributed Computing}, 29(3):187--205, 2016.

\bibitem{Chal}
J{\'{e}}r{\'{e}}mie Chalopin, Emmanuel Godard, and Antoine Naudin.
\newblock Anonymous graph exploration with binoculars.
\newblock In {\em Distributed Computing - 29th International Symposium, {DISC}
  2015}, pages 107--122, 2015.

\bibitem{CicerSurvey}
Serafino Cicerone, Gabriele~Di Stefano, and Alfredo Navarra.
\newblock Asynchronous robots on graphs: Gathering.
\newblock In {\em Distributed Computing by Mobile Entities, Current Research in
  Moving and Computing}, pages 184--217. Springer, 2019.

\bibitem{DefagoSurvey}
Xavier D{\'{e}}fago, Maria Potop{-}Butucaru, and S{\'{e}}bastien Tixeuil.
\newblock Fault-tolerant mobile robots.
\newblock In {\em Distributed Computing by Mobile Entities, Current Research in
  Moving and Computing}, pages 234--251. Springer, 2019.

\bibitem{dieudonne2014gathering}
Yoann Dieudonn{\'e}, Andrzej Pelc, and David Peleg.
\newblock Gathering despite mischief.
\newblock {\em ACM Transactions on Algorithms (TALG)}, 11(1):1, 2014.

\bibitem{Fisch}
Matthias Fischer, Daniel Jung, and Friedhelm {Meyer auf der Heide}.
\newblock Gathering anonymous, oblivious robots on a grid.
\newblock In {\em 13th International Symposium on Algorithms and Experiments
  for Wireless Sensor Networks, {ALGOSENSORS} 2017}, pages 168--181, 2017.

\bibitem{FloccSurvey}
Paola Flocchini.
\newblock Gathering.
\newblock In {\em Distributed Computing by Mobile Entities, Current Research in
  Moving and Computing}, pages 63--82. Springer, 2019.

\bibitem{MobileRobotsSurveyBook}
Paola Flocchini, Giuseppe Prencipe, and Nicola Santoro, editors.
\newblock {\em Distributed Computing by Mobile Entities, Current Research in
  Moving and Computing}.
\newblock Springer, 2019.

\bibitem{hir2020gather}
Jion Hirose, Junya Nakamura, Fukuhito Ooshita, and Michiko Inoue.
\newblock Gathering with a strong team in weakly byzantine environments.
\newblock {\em CoRR}, abs/2007.08217, 2020.

\bibitem{Hsiang}
Tien{-}Ruey Hsiang, Esther~M. Arkin, Michael~A. Bender, S{\'{a}}ndor~P. Fekete,
  and Joseph S.~B. Mitchell.
\newblock Algorithms for rapidly dispersing robot swarms in unknown
  environments.
\newblock In {\em Fifth International Workshop on the Algorithmic Foundations
  of Robotics, {WAFR} 2002}, pages 77--94, 2002.

\bibitem{Oosh}
Fukuhito Ooshita, Ajoy~K. Datta, and Toshimitsu Masuzawa.
\newblock Self-stabilizing rendezvous of synchronous mobile agents in graphs.
\newblock In {\em Stabilization, Safety, and Security of Distributed Systems -
  19th International Symposium, {SSS} 2017}, pages 18--32, 2017.

\bibitem{PelcCrash}
Andrzej Pelc.
\newblock Deterministic gathering with crash faults.
\newblock {\em Networks}, 72(2):182--199, 2018.

\bibitem{PelcSurvey}
Andrzej Pelc.
\newblock Deterministic rendezvous algorithms.
\newblock In {\em Distributed Computing by Mobile Entities, Current Research in
  Moving and Computing}, pages 423--454. Springer, 2019.

\bibitem{Tsuch}
Masashi Tsuchida, Fukuhito Ooshita, and Michiko Inoue.
\newblock Byzantine-tolerant gathering of mobile agents in arbitrary networks
  with authenticated whiteboards.
\newblock {\em {IEICE} Transactions}, 101-D(3):602--610, 2018.

\end{thebibliography}
\end{document}